\documentclass[conference]{IEEEtran}

\usepackage[utf8]{inputenc} 
\usepackage[T1]{fontenc}
\usepackage{url}
\usepackage{ifthen}
\usepackage{cite}
\usepackage[cmex10]{amsmath} 

\usepackage{tikz}
\usepackage[utf8]{inputenc}
\usepackage{bbold}
\usepackage{mathrsfs}
\usepackage{amsmath} \usepackage{amsthm} \usepackage{amsfonts} \usepackage{amssymb} 
\usepackage{epstopdf}
\usepackage{graphicx}
\input{xypic}
\usepackage{bbm}

\usepackage{enumerate}

\newtheorem{theorem}{Theorem}

\newtheorem{lem}{Lemma}
\newtheorem{corollary}{Corollary}

\newtheorem{definition}{Definition}

\theoremstyle{definition}
\newtheorem{example}{Example}

\newtheorem*{prob*}{Problem}

\theoremstyle{remark}
\newtheorem{remark}{Remark}

\global\long\def\NN{\mathbb{N}}

\global\long\def\EE{\mathbb{E}}
\global\long\def\PP{\mathbb{P}}

\global\long\def\11{\mathbb{1}}

\usepackage{multirow}

\begin{document}
\title{Bounds on the Effective-length of Optimal Codes for Interference Channel with Feedback}
\author{
 \IEEEauthorblockN{Mohsen Heidari}
  \IEEEauthorblockA{EECS Department\\University of Michigan\\ Ann Arbor,USA \\
    Email: mohsenhd@umich.edu} 

\and
\IEEEauthorblockN{ Farhad Shirani}
\IEEEauthorblockA{ECE Department\\
New York University\\
New York, New York, 11201\\
Email: fsc265@nyu.edu}
\and

  \IEEEauthorblockN{S. Sandeep Pradhan}
  \IEEEauthorblockA{EECS Department\\University of Michigan\\ Ann Arbor,USA \\
    Email: pradhanv@umich.edu}
}
\IEEEoverridecommandlockouts


\maketitle



\begin{abstract}
In this paper, we investigate the necessity of finite blocklength codes in distributed transmission of independent message sets over channels with feedback. Previously, it was shown that finite effective length codes are necessary in distributed transmission and compression of sources.
We provide two examples of three user interference channels with feedback where codes with asymptotically large effective lengths are sub-optimal. As a result, we conclude that coded transmission using finite effective length codes is necessary to achieve optimality. We argue that the sub-optimal performance of large effective length codes is due to their inefficiency in preserving the correlation between the inputs to the distributed terminals in the communication system. This correlation is made available by the presence of feedback at the terminals and is used as a means for coordination between the terminals when using finite effective length coding strategies.
\end{abstract}

\section{Introduction}
Most of the coding strategies developed in
information theory are based on random code ensembles which are constructed using independent and identically
distributed (IID) sequences of random variables \cite{HK,BT,Marton,ZB}. The codes associated with different terminals in the network are
mutually independent. Moreover, the blocklengths associated with these codes are asymptotically large. This allows the application of
the laws of large numbers and concentration of measure theorems when analyzing the performance of coding strategies; and leads to characterizations of their achievable regions in terms of information quantities that
are the functionals of the underlying distribution used to construct the codes. These characterizations are often called single-letter characterizations. Although the original problem is to optimize the performance of codes with asymptotically
large blocklengths, the solution is characterized by a functional (such as mutual information) of just
one realization of the source or the channel under consideration.  
It is well-known that unstructured random codes with asymptotically large blocklength can be used to achieve optimality in terms of achievable rates in point-to-point communications. In fact, it can be shown that large blocklength codes are necessary to approach optimal performance. At a high level, this is due to the fact that the efficiency of fundamental tasks of communication such as covering and packing increases as the input dimension
is increased \cite{csiszarbook}.

In network communication, one needs to (a) remove
redundancy among correlated information sources \cite{BT,ZB} in a distributed manner in the source coding
problems, and (b) induce redundancy among distributed terminals to facilitate \cite{HK,Marton} cooperation among them.
For example, in the network source coding problems such as distributed source coding and multiple description coding, the objective is to exploit the statistical correlation of the distributed information sources. Similarly, in the network channel coding problems, such
as the interference channels and broadcast channels, correlation of information among different terminals
is induced for better cooperation among them. At a high level, in addition to the basic objectives of efficient packing and covering at every terminal, the network coding strategies need to exploit statistical correlation among distributed information
sources or induce statistical correlation among information accessed by terminals in the network.

Witsenhausen \cite{ComInf2} and Gacs-Korner \cite{ComInf1} made the observation that distributed processing of pairs of sequences of random variables leads to outputs which are less correlated than the original input sequences. In the network communications context, this implies that the outputs of encoding functions at different terminals  in a network are less correlated with each other than the original input sequences. 
 In \cite{fin1,fin2}, we built upon these observations and showed that the correlation between the outputs of pairs of encoding functions operating over correlated sequences is inversely proportional to the effective length of the encoding functions. 
Based on these results, it can be concluded that while random unstructured coding strategies with asymptotically large blocklengths are efficient in performing the tasks of covering and packing, they are inefficient in facilitating coordination between different terminals. 
Using these results, we showed that finite effective codes are necessary to achieve optimality in various setups involving
the transmission of correlated sources. Particularly, we showed that the effective length of optimality achieving codes is bounded from above in the distributed source coding problem as well as the problem of transmission of correlated sources over the multiple access channel (MAC) and the interference channel (IC) \cite{fin1,FinRD}.

So far, all of the results showing the necessity of finite effective length codes pertain to situations involving the distributed transmission of sources over channels and distributed compression of sources. However, the question of whether such codes are necessary in multi-terminal channel coding has remained open. The reason is that the application of the results in \cite{fin1,fin2} requires the presence of correlated inputs in different terminals of the network. In the case of distributed processing of sources, such correlation is readily available in the form of the distributed source. Whereas, in distributed transmission of independent message it is unclear how such a correlation can be created and exploited.  In this work, we argue that in channel coding with feedback, correlation is induced because of the feedback link. More precisely, the feedback sequence at one terminal is correlated with the message set in the other terminal. In order to exploit this correlation efficiently, finite effective length codes are necessary. The contributions of this paper can be summarized as follows. We provide two examples of interference channels with feedback where finite effective length codes are necessary to approach optimality. For each of these examples, we provide an outer bound on the achievable region as a function of the effective-length of the encoding functions used at the transmitters. Furthermore, we use finite effective length codes to prove the achievability of certain rate vectors which lie outside of the outer bound when the effective length is large. The combination of these two results shows that in these examples any coding strategy which uses encoding functions with asymptotically large effective lengths is sub-optimal. 

The rest of the paper is organized as follows: In Section II we introduce the problem formulation. Section III provides the prior results which are used in this paper. Section IV explains our main results. Finally, section V concludes the paper. 

\section{Definitions and Model}
\subsection{Notations}
Random variables are denoted using capital letters such as $X, Y$. The random vector $(X_1, X_2, ..., X_n)$ is represented by $\underline{X}^n$. Similarly, we use underline letters to denote vectors of numbers and functions. For shorthand, vectors are sometimes represented using underline letters without any superscript such as $\underline{X}, \underline{f}$, and $\underline{a}$. Calligraphic letters such as $\mathcal{C}$ and  $\mathcal{M}$  are used to represent sets. 
%

\subsection{Model}
The problem of Interference Channel with Feedback (IC-FB) is studied in  \cite{IC-FB0Tunietti} and \cite{Kramer-ICFB}.   A three-user interference channel with generalized feedback (IC-FB) is characterized by three input alphabets $(\mathcal{X}_1, \mathcal{X}_2, \mathcal{X}_3)$, three output alphabets $(\mathcal{Y}_1, \mathcal{Y}_2,\mathcal{Y}_3),$ three feedback alphabets $(\mathcal{Z}_1, \mathcal{Z}_2, \mathcal{Z}_3)$, and  transition probability distributions $(Q_{\underline{Y}|\underline{X}},P_{\underline{Z}|\underline{Y}})$. We assume that all the alphabets are finite and that the channel is memoryless. Let $\underline{x}_i^n, \underline{y}_i^n, \underline{z}_i^n, i\in [1,3],$ be the channel inputs, outputs and the channel feedback  after $n$ uses of the channel, respectively. The memoryless property implies that:
\begin{align*}
p(y_{j,n}, z_{j,n}, & j\in [1,3]  | \underline{y}_i^{n-1},\underline{z}_i^{n-1}, \underline{x}_i^n, i\in [1,3])\\
&=Q_{\underline{Y}|\underline{X}}(y_{1,n},y_{2,n},y_{3,n}|x_{1,n},x_{2,n},x_{3,n})\\
&\times P_{\underline{Z}|\underline{Y}}(z_{1,n},z_{2,n},z_{3,n}|y_{1,n},y_{2,n},y_{3,n})
\end{align*}
 In the three user IC-FB, there are three transmitters and three receivers. The $i$th transmitter $i\in [1,3],$ intends to transmit the message index $W_i$ to the $i$th receiver. It is also assumed that the feedback $Z_i, i\in [1,3]$ is causally available at transmitter $i$ with one unit of delay. An example of such setup is depicted in Figure \ref{fig: IC diagram}. In this figure, $Z_2$ is trivial, and $P_{Z_1,Z_3|Y_1,Y_2,Y_3}=P_{Z_1|Y_1}P_{Z_3|Y_3}$ (i.e. the second transmitter does not receive any feedback). 

\begin{figure}[!t]
\centering
\includegraphics[scale=0.6]{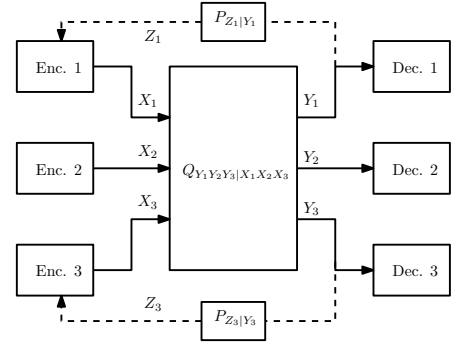}
\caption{An instance of the three-user IC with generalized feedback. Here transmitters 1 and 3 receive noisy feedback, whereas transmitter 2 does not receive feedback.}
\label{fig: IC diagram}
\end{figure}

For positive integers $M_1,M_2,M_3$ and $N$ are arbitrary positive integers.
\begin{definition} \label{def: coding scheme}
A $(M_1,M_2,M_3,N)$ feedback-block-code for the three user IC-FB consists of 
\begin{itemize}
\item Three sets of messages $\mathcal{M}_i=\{1,2, ..., M_i\}, i\in [1,3].$  
\item Three sequences of encoding functions 
$$f_{i,n}: \mathcal{M}_i \times \mathcal{Z}_i^{n-1} \rightarrow \mathcal{X}_i, \quad 1\leq n\leq N,$$
\item Three decoding functions:
\begin{align*}
g_i: \mathcal{Y}_i^N \rightarrow \mathcal{M}_i, \quad i\in [1,3].
\end{align*}
\end{itemize} 
\end{definition}
The message for transmitter $i$ is denoted by a random variable $W_i$. It is assumed that the messages $W_i$ are mutually independent and uniformly distributed on $\mathcal{M}_i, i\in [1,3].$ 
The output of the $i$th transmitter at the $n$th use of the channel is denoted by $X_{i, n}=f_{i,n}(W_i, Z_i^{n-1})$. The rate-triple of an $(M_1,M_2,M_3,N)$ code is defined as $R_i=\frac{\log M_i}{N}, i\in [1,3]$. Let $\widehat{W}_i, i\in [1,3],$ be the decoded message at the receiver $i$. Then, the probability of the error is defined as 
\begin{align*}
P_e \triangleq P((W_1,W_2,W_3) \neq (\widehat{W}_1, \widehat{W}_2, \widehat{W}_3)).
\end{align*}
For this problem, at time $n$, each transmitter can chose an encoding function randomly using a probability measure defined over the set of all encoding functions $\underline{f}_i^N$ described in Definition \ref{def: coding scheme}. The following defines a randomized coding strategy.

\begin{definition} 
A $(M_1,M_2,M_3,N)$-randomized coding strategy is characterized by a probability measure $\mathsf{P}_N$ on the set of all functions $(\underline{f_i}^N), i\in [1,3]$ described in Definition \ref{def: coding scheme}. 
\end{definition}

Next, we define the achievable region for the three user IC-FB. 

\begin{definition} 
For $\epsilon>0$, a rate-triple $(R_1,R_2,R_3)$ is said to be $\epsilon$-achievable by a  feedback-block-code with parameters $(M_1,M_2,M_3,N)$, if the following conditions are satisfied.
$$P_e \leq \epsilon, \quad  \quad \frac{1}{N}\log M_i \geq R_i-\epsilon, ~ i\in [1,3].$$
\end{definition}

\begin{definition} 
For $\epsilon>0$, a rate-triple $(R_1,R_2,R_3)$ is said to be $\epsilon$-achievable by a $(M_1,M_2,M_3,N)$-randomized coding strategy with probability measure $\mathsf{P}$, if, with probability one with respect to $\mathsf{P}$, there exists a feedback-block-code for which $(R_1,R_2,R_3)$ is $\epsilon$-achievable.
\end{definition}

\begin{definition} 
For $\epsilon>0$, a rate-triple $(R_1,R_2,R_3)$ is said to be $\epsilon$-achievable, if there exist a $(M_1,M_2,M_3,N)$ feedback-block-code (randomized coding strategy) for which $(R_1,R_2,R_3)$ is $\epsilon$-achievable. 
\end{definition}

\begin{definition} 
A rate-triple $(R_1,R_2,R_3)$ is said to be achievable, if it is $\epsilon$-achievable for any $\epsilon>0$. Given an IC-FB, the set of all achievable rate-triples is called the feedback-capacity.
\end{definition}

%

\section{Background and Prior Results}
In this section, we summarize the results in \cite{fin1} on the correlation between the outputs of Boolean functions of pairs of sequences of random variables. These results are used in the next section to prove the necessity of finite effective length codes. 

 \begin{definition}
  $(X,Y)$ is called a pair of DMS's if we have $P_{X^n,Y^n}(x^n,y^n)=\prod_{i\in[1,n]}P_{X_i,Y_i}(x_i,y_i), \forall n\in \mathbb{N}, x^n\in \mathcal{X}^n,y^n\in \mathcal{Y}^n$, where $P_{X_i,Y_i}=P_{X,Y}, \forall i\in [1,n]$, for some joint distribution $P_{X,Y}$.
\end{definition}

\begin{definition}
 A Binary-Block-Encoder (BBE) is characterized by the triple $(\underline{e},\mathcal{X},n)$, where $\underline{e}$ is a mapping $\underline{e}:\mathcal{X}^n\to \{0,1\}^n$, $\mathcal{X}$ is a finite set, and $n$ is an integer.  
\end{definition}

\begin{definition}
 For a BBE $(\underline{e},\mathcal{X},n)$ and DMS X,  let $P\left(e_i(X^n)=1\right)=q_i$. 
For each Boolean function $e_i, i\in [1,n]$, the real-valued function corresponding to $e_i$ is defined as follows:
\begin{align}
\tilde{e}_i(X^n)=    \begin{cases}
      1-q_i, & \qquad  \text{if } e_i(X^n)=1, \\
      -q_i. & \qquad\text{otherwise}.
    \end{cases}
\end{align}
\end{definition}

\begin{definition}
 For a BBE $(\underline{e},\mathcal{X},n)$, define the decomposition 
  $\tilde{e}=\sum_{\mathbf{i}}\tilde{e}_{\mathbf{i}}$, where $\tilde{e}_{\mathbf{i}}=\mathbb{E}_{X^n|X_{\mathbf{i}}}(\tilde{e}|X_{\mathbf{i}})-\sum_{\mathbf{j}< \mathbf{i}} \tilde{e}_{\mathbf{j}}$. Then, $\tilde{e}_{\mathbf{i}}$ is the component of $\tilde{e}$ which is only a function of $\{X_{i_j}|i_j=1\}$. The collection $\{\tilde{e}_{\mathbf{i}}|\sum_{j\in[1,n]}i_j=k\}$, is called the set of k-letter components of $\tilde{e}$. \label{Rem:Dec}
\end{definition}


\begin{definition}
 For a function $e: \mathcal{X}^n \to \{0,1\}$, with real decomposition vector $(\tilde{e}_{\mathbf{i}})_{\mathbf{i}\in \{0,1\}^n}$, the dependency spectrum is defined as the vector $(\mathbf{P}_{\mathbf{i}})_{\mathbf{i}\in \{0,1\}^n}$ of the variances, where $\mathbf{P}_{\mathbf{i}}=Var(\tilde{e}_{\mathbf{i}}), \mathbf{i}\in \{0,1\}^n$. The effective length is defined as the expected value $\widetilde{\mathbf{L}}=\frac{1}{n}\sum_{\mathbf{i}\in \{0,1\}^n}w_H(\mathbf{i})\cdot \mathbf{P}_{\mathbf{i}}$, where $w_H(\cdot)$ is the Hamming weight.
 \end{definition}

\begin{lem}\label{lem: single letter codes correlation }
 Let $\psi\triangleq \sup(E(e(X)f(Y))$, where the supremum is taken over all single-letter functions $e:\mathcal{X}\to \mathbb{R}$, and $f:\mathcal{Y}\to\mathbb{R}$ such that $h(X)$ and $g(Y)$ have unit variance and zero mean. 
  the following bound holds:
 \begin{align*}
&2\sqrt{\sum_{\mathbf{i}}\mathbf{P}_{\mathbf{i}}}\sqrt{\sum_{\mathbf{i}}\mathbf{Q}_{\mathbf{i}}}-2\sum_{\mathbf{i}}C_\mathbf{i}\mathbf{P}_{\mathbf{i}}^{\frac{1}{2}}\mathbf{Q}_{\mathbf{i}}^{\frac{1}{2}} 
\leq  P(e(X^n)\neq f(Y^n))
\\&\leq 1- 2\sqrt{\sum_{\mathbf{i}}\mathbf{P}_{\mathbf{i}}}\sqrt{\sum_{\mathbf{i}}\mathbf{Q}_{\mathbf{i}}}+2\sum_{\mathbf{i}}C_\mathbf{i}\mathbf{P}_{\mathbf{i}}^{\frac{1}{2}}\mathbf{Q}_{\mathbf{i}}^{\frac{1}{2}} 
,
\end{align*}
 where 1) $C_{\mathbf{i}}\triangleq  \psi^{N_\mathbf{i}}$, 2) $\mathbf{P}_{\mathbf{i}}$ is the variance of $\tilde{e}_{\mathbf{i}}$, 3) $\underline{\tilde{e}}$ is the real function corresponding to $\underline{e}$, 4) $\mathbf{Q}_{\mathbf{i}}$ is the variance of $\tilde{f}_{\mathbf{i}}$, and 5) $N_{\mathbf{i}}\triangleq w_H(\mathbf{i})$.
 \label{th:sec3_non_bin}
\end{lem}

\begin{remark}
 The value $C_{\mathbf{i}}$ is decreasing in $N_{\mathbf{i}}$. So, $ P(e(X^n)\neq f(Y^n))$, is maximized when most of the variance $\mathbf{P}_{\mathbf{i}}$ is distributed on $\tilde{e}_\mathbf{i}$ which have lower $N_{\mathbf{i}}$ (i.e. operate on smaller blocks). This implies that encoding functions with smaller effective-lengths can have higher correlation between their outputs.
  \end{remark}

\section{Main Results}\label{sec: main results}
In this section, we introduce two examples of three user IC-FBs where finite effective length codes are necessary to approach optimality. 
\begin{example} \label{ex: IC-FB}

Consider the setup shown in Figure \ref{fig: IC-FB-Exp}. Here, $(X_{11}, X_{12}), X_2, (X_{32}, X_{33})$ are the outputs of the $i$th Encoder $i\in [1,3]$, respectively. The channel outputs $Y_1$, $(Y_{2}, Y'_{2})$, and $Y_{3}$ are received at decoders 1,2 and 3, respectively. The channel corresponding to the transition probability $P_{Y'_2| X_{12} X_{32}}$ is described by the following relation:
\begin{align*}
Y'_2=X_{12}+N_{\delta} +(X_{12} \oplus X_{32}) \wedge E,
\end{align*}
where $N_\delta$ and $E$ are independent Bernoulli random variables with $P(N_\delta=1)=\delta$ and $P(E=1)=\frac{1}{2}$. Also, the random variables $N_{\epsilon}$, and $N_{p}$ in Figure \ref{fig: IC-FB-Exp} are Bernoulli random variables with $P(N_\epsilon=1)=\epsilon$ and $P(N_p=1)=p$, respectively. The variables $N_\delta,E,N_\epsilon$ and $N_P$ are mutually independent.
In this setup feedback is only available at encoder $1$ and $3$. The feedback at the first transmitter is $Z_1=Y_{1}$ with probability one. The feedback at the third transmitter is $Z_3=Y_3$ with probability one. In other words the two transmitters receive noiseless feedback.
\end{example}

\begin{figure}[!t]
\centering
\includegraphics[scale=0.6]{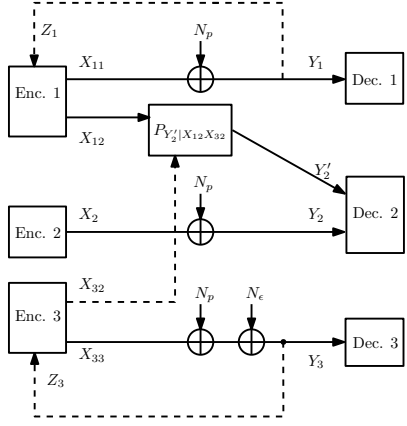}
\caption{The diagram of the IC-FB given in Example 1. In this setup, $Z_1$ is the feedback at Transmitter 1, and $Z_3$ is the feedback at transmitter 3.  }
\label{fig: IC-FB-Exp}
\end{figure}

The following provides an outer bound on the achievable rates after $n$ channel uses. The bound is provided as a function of the average probability of agreement between the encoder outputs $X_{12}$ and $X_{32}$.

\begin{theorem}\label{thm: Exp 1 converse}
Any $(M_1,M_2,M_3,n)$-randomized coding strategy for the channel in Example \ref{ex: IC-FB} achieves a rate vector $(R_1, R_2,R_3)$ satisfying the following inequalities: 
\begin{align*}
R_1 &\leq 1-h(p),\\
 R_2 &\leq 1-\Big|h_b(p)-(1-h_b(\delta))(\frac{1}{n}\sum_{i=1}^nP(X_{12,i}=X_{32,i})\Big|^+),\\
  R_3& \leq 1-h_b(p*\epsilon).
\end{align*}

\end{theorem} 
\begin{proof}
The proof is given in Appendix \ref{appx: proof of thm exp 1 converse}.
\end{proof}
\begin{corollary}
Define the set $\mathcal{R}^*$ as the union of all rate-triples $(R_1, R_2,R_3)$ such that
\begin{align*}
R_1 &\leq 1-h_b(p), \\
R_2 &\leq 1-|h_b(p)-(1-h_b(\delta))|^+,\\
R_3 &\leq 1-h_b(p*\epsilon).
\end{align*}
Then, the feedback-capacity of the channel in Example \ref{ex: IC-FB} is contained in $\mathcal{R}^*$.
\end{corollary}

\begin{corollary}
Suppose $p$ and $\delta$ are such that  $h(p) \leq 1-h(\delta)$, and $\epsilon=0$. Then the feedback-capacity of Example \ref{ex: IC-FB} is characterized by the following 
\begin{align*}
R_1 \leq 1-h(p), R_2\leq 1, R_3\leq 1-h(p).
\end{align*}
\end{corollary}

\begin{proof}
The converse follows by Theorem \ref{thm: Exp 1 converse}. For the achievability, we use standard Shannon random codes at encoder $1$ and $3$. Then, the rates $R_1\leq 1-h(p)$, and $R_3\leq 1-h(p)$ are achievable. Because of the feedbacks $Z_1, Z_3$, the noise  $N_p$ is available at encoder $1$ and $3$. Transmitters one and three send $N_p$ to receiver $2$. We require $N_p$ to be decoded at receiver $2$ losslessly. Consider a good source-channel code for transmission of $N_p$ over a Binary Symmetric Channel with noise bias $\delta$. We use this codebook both at encoder $1$ and $3$. Since the source $N_p$ and the codebook are available at encoder $1$ and $3$, then $X_{12}=X_{32}$ with probability one. As a result, the channel $P_{Y'2| X_{12}X_{32}}$ becomes a binary symmetric channel with bias $\delta$. Therefore, as $h_b(p)\leq 1-h(\delta)$ then $N_p$ is reconstructed at receiver 2 without any noise. By subtracting $N_p$ from $Y_2$ the channel from $X_2$ to $Y_2$ becomes a noisless channel. Thus, $R_2=1$ is achievable. 
\end{proof}

\begin{lem}\label{lem: IC-FB exp 1 cap continuouty}
Let $\mathcal{C}_\epsilon$ denote the feedback-capacity region of the IC-FB in Example \ref{ex: IC-FB}. For any $(R_1,R_2,R_3)\in \mathcal{C}_0$, there exists a continuous function  $\zeta(\epsilon)$ such that for sufficiently small $\epsilon>0$ the rate-triple $(R_1-\zeta(\epsilon), R_2-\zeta(\epsilon), R_3-\zeta(\epsilon))\in \mathcal{C}_\epsilon$, where $\zeta(\epsilon)\rightarrow 0$, as $\epsilon \rightarrow 0$. 
\end{lem}
\begin{proof}
The proof is given in Appendix \ref{sec: proof of lemma capacity continuouty}.
\end{proof}

\begin{theorem}\label{thm: IC-FB exp 1 single letter suboptimal}
There exist $\gamma>0$ and $\epsilon >0$,  
such that for any coding strategy achieving the rate-triple $(1-h_b(p), 1-\gamma, 1-h_b(p))$ the effective length of the encoding functions producing $X_{12}$ and $X_{32}$ are bounded from above by a constant. Furthermore, the effective length is greater than 1 (i.e. uncoded transmission is not optimal).
\end{theorem}

\begin{proof}[proof outline]
 From Theorem \ref{thm: Exp 1 converse} the following upper-bound holds for $R_2$.
\begin{align*}
R_2 \leq 1-h_b(p)(1-\frac{1}{n}\sum_{i=1}^nP(X_{12,i}=X_{32,i}))
\end{align*} 
Therefore, it is required that  
\begin{align*}
\frac{1}{N}\sum_{i=1}^N P(X_{12,i}= X_{32,i})\approx 1. 
\end{align*}
This implies that $\forall n\in\mathbb{N}$, $P(X_{12,n}= X_{32,n})\approx 1$. However, by Lemma \ref{lem: single letter codes correlation }, 
this requires that the effective length be bounded from above. If the effective length is equal to 1, then $\mathbf{P}_{\mathbf{i}_n}\approx 1$ for all $n\in \mathbb{N}$, this implies that  $P(F_{1,n}(Z_1^{n-1})= Z_{1, n-1})\approx 1$. Thus, $P(Y'_{2,n} = N_p+N_\delta)\approx 1 $.
 However,
\begin{align*}
\frac{1}{N}H(N_{p,N_0}^N| Y^{'N}_{2,N_0}) &\approx (1-\frac{N_0}{N})( 2h_b(p)-h_b(p*p))\\
& \approx 2h_b(p)-h_b(p*p) >0
\end{align*} 
As a result it is not possible to reconstruct $N_p$ at the decoder losslessly. More precisely, $R_2 \lesssim 1+h_b(p*p)-2h_b(p)$ This contradicts with $R_2\approx 1$.
%
%

\end{proof}

\subsection{The Second Example}
In this subsection, we provide another example to illustrate more the necessity of coding strategy with effective finite length for communications over IC-FB.
\begin{example}\label{ex: IC-FB ex 2}
Consider the IC shown in \ref{fig: IC-FB exp 2 diagram }. The outputs of encoder 1 are denoted by $(X_{11}, X_{12})$, the output of encoder 2 is $X_2$, and the outputs of encoder 3 are $(X_{32}, X_{33},X'_{33} )$. In this setup, $Z_1$ and $Z_2$ represent the feedback available at encoder 1 and encoder 3, respectively. All the inputs alphabets in this channel are binary. All the output alphabets are binary; except $Y_1$ which a ternary. In this setup $N_1$, $N_3, N_{\delta}, N_{\epsilon}$ and $E$ are mutually independent Bernoulli random variables with parameter $p_1, p_3, \delta, \epsilon$, and $1/2$, respectively. Finally, it is assumed that  $p_1, p_3, \delta, \epsilon <1/2$. 
\begin{figure}[hbtp]
\centering
\includegraphics[scale=0.6]{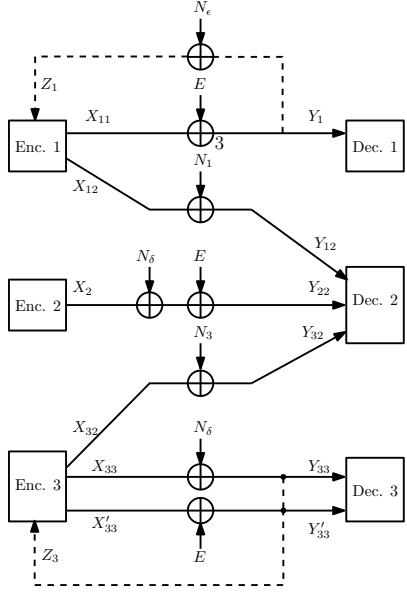}
\caption{The diagram of the IC-FB in Example 2. In this setup, $Z_1$, the feedback at Transmitter 1, is a noisy version of $Y_1$.}
\label{fig: IC-FB exp 2 diagram }
\end{figure}

\end{example}

We first study the case in which $\epsilon=0.$ The following lemma provides an achievable rates for this example.

\begin{lem}\label{lem: achievability exp 2}
For $\epsilon=0$ in the setup given in Example \ref{ex: IC-FB ex 2}, the rate-triple $(\log 3-1, 1-h_b(d), 1-h_b(\delta))$ is achievable, where $$d=h_b^{-1}(|h_b(p_1 * \delta)+h(p_3)-1|^+).$$
\end{lem}
\begin{proof}[proof outline]
The bounds on $R_1, R_3$ follow from the standard arguments as in point-to-point channel coding (Fano's inequality). Next, we show the bound on $R_2$. Upon receiving $Z_1$, the first encoder recovers $E$. The third  encoder receives $Z_3$ and recovers $(E, N_\delta)$. Encoder 1 and $3$ employ a source-channel coding scheme to encode the sources $E, N_\delta$ such that decoder 2 be able to reconstruct $E+N_\delta$ within a Hamming distortion $d$. This is a similar problem to the Common-Bit One-Help-One Problem introduced in \cite{wagner_dist}. Using the results from \cite{wagner_dist} (Theorem 3), we can show that decoder 2 is able to reconstruct $E+N_\delta$ within a Hamming distortion $d$, if the bounds 
\begin{align*}
R_{32} &\geq h_b(\Delta*\delta)-h_b(d), \text{and} \quad R_{12} \geq 1-h_b(\Delta),
\end{align*} 
hold for some $0 \leq \Delta \leq 1/2$. From standard channel coding arguments the transmitted codewords from encoder 1 and 3 are decoded at receiver 2 with small probability of error, if $$R_{12}\leq 1-h_b(p_1)-\zeta, \quad \quad  R_{32}\leq 1-h_b(p_3)-\zeta,$$ where $\zeta>0$ is a sufficiently small number.  Finally the proof follows by setting $\Delta\approx p_1$, and $d$ as in the statement of the lemma. 

\end{proof}

For the case when $\epsilon>0$, there is no common information between encoder $1$ and 3. From the discontinuity argument as in \cite{wagner_dist}, we can show that the minimum distortion level $d$ is discontinuous as a function of $\epsilon$. This implies that the achievable rates using single letter coding scheme strictly decreases comparing to the case when $\epsilon=0$. Hence, there exists a $\gamma>0$ such that any rate-triple $(R_1,R_2,R_3)$ with  $R_2 >  1-h_b(\delta)-\gamma$ is not achievable using single-letter coding strategies, where $d$ is as in Lemma \ref{lem: achievability exp 2}. More precisely, the following Lemma holds.

\begin{lem}
There exist $\gamma>0$ and $\epsilon >0$, such that for any coding strategy achieving the rate-triple $(\log_3 -1, 1-h_b(d)-\gamma, 1-h_b(\delta))$ the effective length of the encoding functions producing $X_{12}$ and $X_{32}$ are bounded from above by a constant. Furthermore, the effective length is greater than 1 (i.e. uncoded transmission is not optimal).
\end{lem}   

\begin{proof}
The proof for this lemma follows from a similar argument as in Theorem \ref{thm: IC-FB exp 1 single letter suboptimal}.
\end{proof}

\section{Conclusion}
We provided two examples of channel coding with feedback over interference networks where finite effective length coding is necessary to achieve optimal performance. We showed that in these examples, optimality achieving coding strategies utilize the feedback available in different terminals to coordinate their outputs. We showed that coding strategies with asymptotically large effective lengths are inefficient in preserving the correlation among their outputs and are hence unable to coordinate their inputs to the channel effectively.

\appendices
\section{Proof of Theorem \ref{thm: Exp 1 converse}}\label{appx: proof of thm exp 1 converse}
\begin{proof}
The bounds $R_1 \leq 1-h_b(p)$ and $R_3 \leq 1-h_b(p)$ follows from standard arguments as in point-to-point channel coding problem. Note that the feedback does not increase the rate of $R_1$ and $R_3$ since these upper-bounds correspond to the point-to-point capacity and feedback does not increase point-to-point capacity. To bound $R_2$ we use Fano's inequality.  Therefore, ;  
\begin{align}
\nonumber nR_2 \leq H(W_2)&\stackrel{(a)}=H(W_2|Y^{'n}_{2})\\
\nonumber &\stackrel{(b)}{\leq} I(W_2; Y^n_{2}|Y^{'n}_{2})+n\zeta_n \\
\nonumber &= H(Y^n_{2}|Y^{'n}_{2})-H(Y^n_{2}|W_2, Y^{'n}_{2})+n\zeta_n \\
\nonumber &\stackrel{(c)}{=} H(Y^n_{2}|Y^{'n}_{2})-H(Y^n_{2}|W_2,X_2^n, Y^{'n}_{2})+n\zeta_n \\
\nonumber &\stackrel{(d)}{\leq} n-H(Y^n_{2}|W_2,X_2^n, Y^{'n}_{2})+n\zeta_n \\
\label{eq:A1}&\stackrel{(e)}{=}n-H(N_p^n|Y^{'n}_{2})+n\zeta_n,
\end{align}
where (a), (c) and (e) follow from the fact that $X^n_{12}, X^n_{32}, Y^{'n}_{2}$ are independent of $W_2$ and that $X_2^n$ is a function of $W_2$ since the second transmitter does not receive feedback. (b) follows from Fano's inequality and (d) follows from the fact that $Y_2$ is binary.

Define the random vector $Z_i, i\in [1,n]$ as the indicator function of the event that $X_{12, i}=X_{32, i}$. Then,
\begin{align*}
H(N_p^n|Y^{'n}_{2}) &\geq H(N_p^n|Y^{'n}_{2}, Z^n)\\
&=\sum_{\underline{z} \in \{0,1\}^n} p(Z^n=\underline{z}) H(N_p^n|Y^{'n}_{2}, \underline{z})
\end{align*}
For the innermost term in the above inequality we have:
\begin{align}
\nonumber H(N_p^n|Y^{'n}_{2}, \underline{z})
&\stackrel{(a)}{=}H(N_p^n|X_{12}^{\underline{z}}\oplus N_\delta^{\underline{z}})\\
\nonumber &=H(N_p^n, X_{12}^{\underline{z}}\oplus N_\delta^{\underline{z}})-H(X_{12}^{\underline{z}}\oplus N_\delta^{\underline{z}})\\
\nonumber &\stackrel{(b)}{\geq} H(N_p^n, X_{12}^{\underline{z}}\oplus N_\delta^{\underline{z}})-w_H(\underline{z})\\
\nonumber &\stackrel{(c)}{\geq} H(N_p^n, N_\delta^{\underline{z}})-w_H(\underline{z})\\
\nonumber &= H(N_p^n)+H(N_\delta^{\underline{z}})-w_H(\underline{z})\\
&= nh_b(p)-w_H(\underline{z})(1-h_b(\delta)),
\label{eq:A2}
\end{align}
where (a) follows from the definition of $Y'_2$, (b) follows from the fact that the binary entropy is upper bounded by one and finally, (c) follows form the fact that $X_{12}$ is independent of $N_{\delta}$. 
Combining equations \eqref{eq:A1} and \eqref{eq:A2}, we get: 
\begin{align*}
H(N_p^n|Y^{'n}_{2}) &\geq n|h_b(p)- (1-h_b(\delta)) \frac{1}{n}\EE[w_H(\underline{z})]|^+
\end{align*}
\end{proof}

\section{Proof of Lemma \ref{lem: IC-FB exp 1 cap continuouty}}\label{sec: proof of lemma capacity continuouty}
\begin{proof}
 We can show that $\mathcal{C}_\epsilon$ is the set of all rate-triples $(R_1,R_2,R_3)$ for which $\exists N \in \NN$ such that
\begin{align*}
R_1& \leq \frac{1}{N} I(X_1^N; Y_1^N)\\
R_3 & \leq \frac{1}{N}I(X_{33}^NX^{'N}_{33}; Y_{33}^NY_{33}^{'N}) \\
R_2 &\leq \frac{1}{N}I(X_2^N; Y_2^N|Y^{'N}_2).
\end{align*}
where the joint distribution of the variables is in some set $\mathcal{P}_\epsilon$. The proof for this statement follows by a converse and an achievability argument. For the converse, we use Fanoe's inequality as in Theorem \ref{thm: Exp 1 converse}. The achievability is straight-forward and follows by employing a multi-letter random coding scheme.  For the case when $\epsilon=0$, and any achievable rate-triple $(R_1,R_2,R_3)$ there exist $N$ and $\gamma>0$ such that 
\begin{align}
R_1& \leq 1-h_b(p)-\gamma\\
R_3 & \leq 1-h_b(p)-\gamma\\\label{eq: bound on R_2 for ep =0}
R_2 &\leq \frac{1}{N}I(X_2^N; Y_2^N|Y^{'N}_2)-\gamma,
\end{align}
where the joint distribution of the random variables involved is denoted by $\PP_0\in \mathcal{P}_0$.
Since the binary entropy function is continuous, then there exists $\zeta(\epsilon)$ such that $1-h_b(p) \leq 1-h_b(p*\epsilon)+\zeta(\epsilon)$.  Next, we show continuity for the right-hand side of \eqref{eq: bound on R_2 for ep =0}. Fix $N, \PP_\epsilon \in \mathcal{P}_\epsilon$, and consider the third inequality in the characterization of $\mathcal{C}_\epsilon$. Note that the only probability distribution depending on $\epsilon$ is $P(y_{3,n}|x_{3,n})$. Since this conditional probability is continuous with $\epsilon$ then so is $\PP_\epsilon$. 
Thus, for any fixed $N$, the bound on $R_2$ in $\mathcal{C}_{\epsilon}$ is continuous as a function function of $\epsilon$. As a result, there exists a function $\zeta'(\epsilon)$ such that the right-hand side of the inequality \eqref{eq: bound on R_2 for ep =0} is upper bounded by $\frac{1}{N}I(X_2^N; Y_2^N|Y^{'N}_2)+\zeta'(\epsilon)-\gamma$ for some joint distribution $\PP_\epsilon \in \mathcal{P}_\epsilon$. As a result, the following bounds hold for $(R_1,R_2,R_3)$.
\begin{align*}
R_1& \leq 1-h_b(p)+\zeta(\epsilon)-\gamma\\
R_3 & \leq 1-h_b(p*\epsilon)+\zeta(\epsilon)-\gamma\\
R_2 &\leq \frac{1}{N}I(X_2^N; Y_2^N|Y^{'N}_2)+\zeta(\epsilon)-\gamma.
\end{align*}
 This implies that there exists $\epsilon>0$ sufficiently small such that $(R_1-\zeta(\epsilon),R_2-\zeta(\epsilon),R_3-\zeta(\epsilon))\in \mathcal{C}_\epsilon$. Thus, we establish the continuity of $\mathcal{C}_\epsilon$ at $\epsilon=0$.  
\end{proof}
%

\end{document}